\documentclass[sigconf]{acmart}
\usepackage[labelformat=simple]{subcaption}
\usepackage{microtype}

\AtBeginDocument{%
  }

\newcommand{\liff}{\leftrightarrow}
\DeclareMathOperator{\MULT}{MULT}
\usepackage{mathtools}

\copyrightyear{2024}
\acmYear{2024}
\setcopyright{acmlicensed}\acmConference[ISSAC '24]{International Symposium on Symbolic and Algebraic Computation}{July 16--19, 2024}{Raleigh, NC, USA}
\acmBooktitle{International Symposium on Symbolic and Algebraic Computation (ISSAC '24), July 16--19, 2024, Raleigh, NC, USA}
\acmDOI{10.1145/3666000.3669712}
\acmISBN{979-8-4007-0696-7/24/07}

\begin{document}

\title{SAT and Lattice Reduction for Integer Factorization}

\author{Yameen Ajani}
\email{ajaniy@uwindsor.ca}
\affiliation{%
  \institution{University of Windsor}
  \streetaddress{401 Sunset Ave}
  \city{Windsor}
  \state{Ontario}
  \country{Canada}
}

\author{Curtis Bright}
\email{cbright@uwindsor.ca}
\orcid{0000-0002-0462-625X}
\affiliation{%
  \institution{University of Windsor}
  \streetaddress{401 Sunset Ave}
  \city{Windsor}
  \state{Ontario}
  \country{Canada}
}

\renewcommand{\shortauthors}{Y.~Ajani and C.~Bright}

\begin{abstract}
The difficulty of factoring large integers into primes is the basis for cryptosystems such as RSA\@.
Due to the widespread popularity of RSA, there have been many proposed attacks on the factorization problem such as side-channel attacks where some bits of the prime factors are available.
When enough bits of the prime factors are known, two methods that are effective at solving the factorization problem are satisfiability (SAT) solvers and Coppersmith's method.
The SAT approach reduces the factorization problem to a Boolean satisfiability problem, while Coppersmith's approach uses lattice basis reduction.
Both methods have their advantages, but they also have their limitations: Coppersmith's method does not apply when the known bit positions are randomized, while SAT-based methods can take advantage of known bits in arbitrary locations, but have no knowledge of the algebraic structure exploited by Coppersmith's method.
In this paper we describe a new hybrid SAT and computer algebra approach to efficiently solve random leaked-bit factorization problems.
Specifically, Coppersmith's method is invoked by a SAT solver to determine whether a partial bit assignment can be extended to a complete assignment.
Our hybrid implementation solves random leaked-bit factorization problems significantly faster than either a pure SAT or pure computer algebra approach.
\end{abstract}



\maketitle

\section{Introduction}
Integer factorization is a well-studied and important problem in the mathematics and computer science community, both because
of its theoretical elegance but also because its difficulty forms the theoretical basis of popular cryptosystems such as RSA.
Given an integer $N$, the factorization problem is to decompose $N$ as a product $N=p_1 \dotsb p_k$ where the~$p_i$
are prime numbers.  Up to ordering of the prime factors~$p_i$ (some of which may appear multiple times)
the factorization is unique---a fact that was essentially shown by Euclid around 300 BC,
though not stated in full completeness until 1801 by Gauss~\cite{Collison1980}.

It is unknown if there exists an algorithm that can factor integers in polynomial time
in the bitlength of $N$, at least on a classical computer.  The fastest general algorithm
discovered to date is the number field sieve~\cite{1993} which heuristically runs in sub-exponential time.
In addition, Shor's algorithm~\cite{Shor1999} is a quantum-based method that can factor composites in polynomial time
subject to the availability of a fault-tolerant quantum computer.
The difficulty of factoring integers---especially semiprimes (numbers with exactly two prime factors)---forms
the basis many cryptosystems currently in wide usage such as RSA.

The most successful methods proposed to solve the factorization problem exploit the algebraic structure inherent in the problem.
A separate approach reduces factoring~$N$ to a Boolean satisfiability (SAT) problem
that when solved reveals a nontrivial factor of~$N$.
In recent years, SAT solvers have achieved great success on many varied kinds of search
problems---there are numerous practical and theoretical problems for which
SAT solvers are the most effective known way of solving the problem~\cite{handbook2021}.
Some difficult mathematical problems---such as the resolution of the Boolean Pythagorean triples
problem~\cite{Heule2016} or the computation of the fifth Schur number~\cite{Heule2018}---have
\emph{only} been solved using SAT solvers.
Unfortunately, for the factoring problem specifically, the SAT approach is dramatically outperformed by
algebraic algorithms~\cite{Mosca2022}. This is not surprising, since although SAT solvers
are great general-purpose search tools, they struggle with problems having a mathematical structure
unknown to the solver~\cite{Bright2022}.

Recently, there have been SAT solvers 
augmented with a programmatic
interface supporting the injection of logical facts as the
solver is running~\cite{Ganesh2012,lipics.sat.2023.8,bright2016mathcheck2}.
Such an approach has successfully resolved mathematical
problems that were beyond the reach of SAT solvers or algebraic methods
alone~\cite{Bright2019}.  For example, progress has been made on certain
mathematical conjectures by extracting mathematical facts from
a computer algebra system (CAS) and programmatically passing them
to a SAT solver as the solver is running~\cite{Zulkoski2016}.
Augmenting a SAT solver in this way can dramatically improve
its effectiveness---intuitively, it is no longer restricted
to reasoning on the level of Boolean logic.
On the other hand, such a solver can also outperform
pure algebraic methods, especially on problems that benefit
from efficient search routines.  Intuitively, this is because
traditionally CASs have not exploited the effective
search-with-learning algorithms developed for SAT solvers~\cite{Abraham2015}.

In the past decade, the line between SAT solving
and computer algebra is starting to blur
with the development of numerous hybrid methods exploiting
SAT solvers in conjunction with computer algebra~\cite{Davenport2020}.
For example, the ``SC-square'' project combines
SAT and computer algebra and has been applied
to fields as diverse as economics, dynamic geometry,
and knot theory~\cite{England2022a}.

\subsection{Our contributions}

In this paper, we introduce a new programmatic SAT method that dramatically
improves the performance of SAT solvers on integer factorization problems
by exploiting algebraic structure of the problem that would otherwise
be hidden from the solver.  More precisely, we employ Coppersmith's
method~\cite{don-cs} for finding small roots of polynomials modulo
a number~$N$ using lattice basis reduction.

Coppersmith's method can factorize
a semiprime~$N$ in polynomial time when either the top half or the
bottom half of the bits of one of its prime factors is known~\cite{DeMicheli2024}.
We exploit the algebraic structure revealed by Coppersmith's method
in the programmatic SAT solver MapleSAT~\cite{DBLP:conf/sat/LiangGPC16}
by querying a computer algebra system
supporting the necessary lattice basis reduction routines.
The information provided by Coppersmith's method is translated into
logical facts that the solver uses to backtrack much earlier
than it otherwise would, dramatically improving the performance
of the solver.

It should be stressed that our approach is not directly competitive with
the best algebraic methods for the integer factorization problem.  However, due
to the practical importance of the factoring problem it has long been
of interest to study weakenings of the factorization problem where
some information about the prime factors are assumed to be known in
advance.  In practice, such information may be leaked through
side-channel attacks (see Section~\ref{sec:sidechannel}).  In our work, we consider random leaked-bit
factorization problems---i.e., where random bits of the prime factors
of the number to factor are known, but the attacker
\emph{has no control over which bits are leaked}.

Although Coppersmith's method requires only half of the bits of
the prime factors to be known (see Section~\ref{sec:fact-coppersmith}),
the method requires the known bits to be \emph{consecutive}---ideally
either the high bits or low bits
of one of the prime factors.  Coppersmith's method can be adapted to work
with multiple chunks of unknown bits, but it is
exponential in the number of chunks~\cite{Herrmann2008}.  Thus, in general
Coppersmith's method does not directly apply when the known bit positions are distributed
uniformly at random.

Conversely, our method takes advantage of known bits
from arbitrary positions but also takes advantage of
the algebraic relationships revealed by Coppersmith's method.
Our results, discussed in Section~\ref{sec:results},
show that our augmented SAT solver can solve some
leaked-bit factorization problems exponentially faster than
an off-the-shelf SAT solver.  It also
outperforms a brute-force approach of trial division
by all factors consistent with the known bits, even if
Coppersmith's method is used to speed up the brute-force
guessing, and can outperform the ``branch and prune''~\cite{HS09}
approach (see Section~\ref{sec:comparison}).  With enough known bits our method
even outperforms the fastest general-purpose factoring
algorithms such as the number field sieve, though we admit this is not really
a fair comparison since the number field sieve seems unable to
exploit known bits and hence is at a disadvantage for the leaked-bit factorization problem
we consider in this paper.
In summary, our method outperforms algebraic methods, pure SAT methods, and
a brute-force + Coppersmith method.

\section{Preliminaries}\label{sec:preliminaries}

In this section we outline the preliminaries needed to understand our approach
for solving random leaked-bit factorization problems.

\subsection{RSA cryptosystem}

RSA is a public-key cryptography system used for signing and encrypting messages 
invented by Rivest, Shamir and Adleman~\cite{rivest1978method}.
As a public-key cryptosystem, RSA uses a public key and private key for message encryption and decryption, respectively.
An RSA user creates a set of two keys, public and private, which are specific to that particular user.
The public key is publicly available and can be used by anyone who wants to send an encrypted message to the user.
The private key is secret (available only to the recipient) and is used to decrypt the encrypted messages received by the user.
RSA public keys typically contain a number~$N$ that is the product of two large primes
and the security of the RSA scheme relies on the practical difficulty of factoring~$N$.

Some common terms used with respect to RSA are
the primes $p$ and $q$, the RSA modulus $N=p\cdot q$, the public exponent $e$, and the private exponent $d$.
The public exponent $e$ is often a fixed size; commonly $e=3$ or $e=65\mathord{,}537$.  The exponent~$e$ must be chosen
to share no common factors with both $p-1$ and $q-1$.  The key parameters are chosen so that the functions
$x\mapsto x^e\bmod N$ and $x\mapsto x^d\bmod N$ are inverses of each other.  If an attacker can factor~$N$,
then they can easily compute the private exponent $d=e^{-1}\bmod(p-1)(q-1)$ via the extended Euclidean algorithm,
and in fact computing~$d$ from $(e,N)$ is polynomial time equivalent to factoring $N$~\cite{May2004}.

\subsection{Side channel attacks}\label{sec:sidechannel}

Side-channel attacks aim to exploit information unintentionally leaked by a computer system or a device.
For example, cold boot attacks are a type of side-channel attack exploiting information remaining in the dynamic random-access memory (DRAM) of a computer system after an attacker cuts the power.
\citeauthor{Halderman:10.1145/1506409.1506429} demonstrate that this remanence effect makes it possible to
recover the contents of a computer's memory with high accuracy
after power has been removed, especially when the DRAM
is subjected to low temperature~\cite{Halderman:10.1145/1506409.1506429}.
Additionally, it was found that bits in DRAM modules tend to decay to a predictable ground state.
For example, the bits holding a private key may be known to decay to~0, not 1.
In this case, after performing a cold-boot attack, any bits that are still 1 are known to have
originally been 1, while bits that are 0 are unknown.
In this way, the known bits of the private key learned by the attacker may be randomly
distributed throughout the key.
Incredibly, experiments show that
disconnecting DRAM and storing it in liquid nitrogen for an hour
resulted in a decay of only 0.13\% of the bits~\cite{Halderman:10.1145/1506409.1506429}.

There are many other avenues from which bits of the private keys may be leaked,
including cache timing attacks on modular exponentiation, and security vulnerabilities
like Heartbleed, Spectre, and Meltdown~\cite{DeMicheli2024}.
Such vulnerabilities typically enable reading arbitrary memory contents (rather than leaking random bits),
though in some cases the attacks may leak only incomplete information.

\subsection{Boolean satisfiability}

The Boolean satisfiability (SAT) problem is to determine whether a formula in Boolean logic has an
assignment to its variables under which the statement becomes true.
If such an assignment exists, the problem is said to be \emph{satisfiable}.
Although SAT is an NP-complete problem and thought to be impractical to solve in the worst case,
in practice there are ``SAT solvers'' that can find satisfying assignments---or
prove the nonexistence of satisfying assignments---for many statements of practical interest.

Most SAT solvers require the input statement to be written in conjunctive normal form~(CNF), i.e.,
a conjunction of \emph{clauses}---clauses being formulae of the form $l_1\lor\dotsb\lor l_k$ where each $l_i$
is a Boolean variable or negated Boolean variable.  One of the most effective solving paradigms
is conflict-driven clause learning~\cite{MarquesSilva}
in which the solver learns new clauses as it searches for a satisfying assignment.

\subsection{SAT \& computer algebra}

Combining SAT with computer algebra systems (CASs) was proposed in 2015 by~\cite{Abraham2015} and~\cite{Zulkoski2015}.
Soon afterwards, the \emph{SC-Square project}~\cite{brahm2017}
started with the aim of facilitating connections between the communities
of satisfiability checking and symbolic computation.
Until that point, the two communities were largely separated, with ``satisfiability checking''
largely focused on search algorithms
and ``symbolic computation'' largely focused on mathematical algorithms.

Many successful applications have arisen as a result of connecting the two fields~\cite{Bright2022,England2022a}.
For example, proving the correctness of multiplier circuits~\cite{Kaufmann2023},
finding new algorithms for matrix multiplication~\cite{Heule2021},
making progress on conjectures in geometric group theory~\cite{Savela},
debugging of digital circuits~\cite{Mahzoon2018},
generating combinatorial objects up to isomorphism~\cite{lipics.sat.2023.14,DBLP:conf/scsquare/LiBG22},
and searching for collisions in hash functions such as
step-reduced SHA-256~\cite{alamgir2024}.

\subsection{Lattices and the LLL algorithm}

A lattice is a discrete and periodic set of points in Euclidean space.
It can be visualized as an infinite grid-like structure where each point is an integer linear combination of a set of linearly independent basis vectors (see Figure~\ref{fig:lattice}).
Lattices have applications in many fields of mathematics and computer science---number theory and cryptography in particular.
In particular, lattices are an essential component of Coppersmith's method that we rely on in our hybrid SAT and computer algebra
factorization approach.

\begin{figure}
\centering
\includegraphics{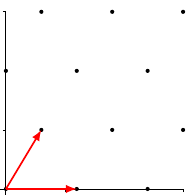}%
\caption{A two dimensional lattice generated by two vectors.}%
\label{fig:lattice}%
\end{figure}

The LLL algorithm is a lattice basis reduction technique used to find short and nearly orthogonal basis vectors for a lattice~\cite{Lenstra1982}.
Given a lattice defined by an $n\times n$ basis matrix, the LLL algorithm runs in polynomial time in~$n$ and finds another
basis of the given lattice where the first vector in the basis has length at most $2^{(n-1)/2}$ times that of the shortest nonzero vector of the lattice.

\subsection{Coppersmith's method}\label{sec:coppersmith}

Coppersmith's method finds small integer roots of a polynomial modulo a given integer~$N$~\cite{don-cs}.
Coppersmith's algorithm runs in polynomial time---even when the factorization of $N$ is unknown---and
it is this property that makes it a useful subroutine in certain relaxations of the factorization problem.

The method exploits a connection between short vectors and polynomials with small coefficients. 
Given a modulus~$N$ and a polynomial $f$ with a small root $x_0$ modulo $N$, Coppersmith's algorithm constructs a lattice 
for which every vector in the lattice corresponds to a polynomial with $x_0$ as a root modulo~$N$.

If a lattice vector is short enough, it will correspond to a polynomial $g$ for which $\lvert g(x_0)\rvert<N$.
Because $g(x_0)\equiv0\pmod N$ by construction, this implies $g(x_0)=0$ and thus $x_0$ is a root of $g$ \emph{over the integers}---not just mod~$N$.
Since the integer roots of a polynomial can be computed in polynomial time~\cite{vonzurGathen2013}, this reduces
the problem of finding the root $x_0\pmod{N}$ to the problem of finding a short vector in Coppersmith's lattice.

\subsection{Factoring with Coppersmith's method}\label{sec:fact-coppersmith}

We summarize Coppersmith's method as used in the factorization context.
For more details, see~\cite{Bos_Stam_2021,HowgraveGraham1997FindingSR}.

We assume that $N=p\cdot q$ is a semiprime;
for example, take $N=16803551=2837\cdot5923$.
Coppersmith's method can factorize~$N$ when either the top
or bottom half of the bits of~$p$ are known,
i.e., at least 50\% of $p$'s bits are known and
these are either $p$'s most significant bits (MSBs)
or least significant bits (LSBs).
%
In the former case, $p$ can be written as $p = \hat{p} + \check{p}$
where $\hat{p}$ is an integer encoding the known high bits as $p$,
and $\check{p}$ is an integer encoding the unknown low bits of~$p$.
As an example (using decimal digits instead of binary digits for simplicity), if $p = 2837$ and $\hat{p} = 2830$, then $\check{p} = 7$.

As described in Section~\ref{sec:coppersmith}, Coppersmith's method
finds small roots $x_0$ of a polynomial $f(x)$ modulo some integer.  
In the factoring application, the modulus used is~$p$. 
Note that $p$ is unknown, but we do know~$N$ (a multiple of~$p$)
which is enough---in this case Coppersmith's method returns the small integer $x_0$ for which
$f(x_0)\equiv 0 \pmod p$, and therefore
$f(x_0)\bmod N$ is divisible by $p$, so $p$ can be extracted by taking a greatest
common divisor between $f(x_0)$ and~$N$.
We take $f(x)=\hat{p}+x$ which has the small root $\check{p}$ modulo~$p$
since $f(\check{p})=\hat{p}+\check{p}\equiv0\pmod{p}$.

Now consider the polynomials
$f(x)$, $xf(x)$, $x^2f(x)$, and the constant polynomial $N$ (note that
indeed $\check{p}$ is a root of each of these polynomials modulo~$p$).
Lattice basis reduction will be applied to the lattice basis
generated by $\{N, f(x), xf(x), x^2f(x)\}$
where a polynomial $a_0+a_1x+a_2x^2+a_3x^3$
is represented by the lattice vector $(a_0,10a_1,100a_2,1000a_3)$ or in general
$(a_0,Xa_1,X^2a_2,X^3a_3)$ where $X$ is a bound on the size of $\check{p}$.
Once a short vector of the lattice is uncovered, integer root detection can reveal
the small root $x_0=\check{p}<10$ from which $p=\hat{p}+x_0$ is uncovered.
It is possible that the polynomial associated with the short vector has other
integer roots, and in that case to reveal~$p$
one should check for \emph{each} root~$x_0$ if $\hat{p}+x_0$ divides $N$.
See Figure~\ref{fig:cs_working} for a diagrammatic working of Coppersmith's method.

\begin{figure}
    \centering
    \includegraphics[width=\linewidth]{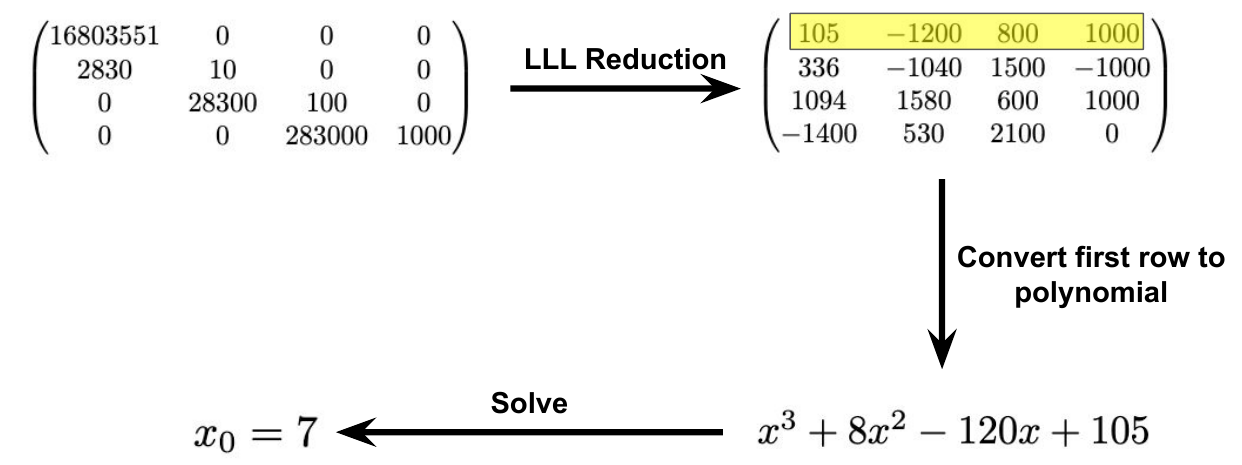}
    \caption{\boldmath Example to demonstrate the working of Coppersmith's method in the case where $N = 16803551$ and $f(x) = 2830 + x$.
    After applying lattice basis reduction, the short polynomial $x^3+8x^2-120x+105$ is discovered which has the integer root $x_0=7$.
    Finally, $f(x_0)=2837$ is a factor of $N$.}%
    \label{fig:cs_working}%
\end{figure}

\subsubsection{Factoring with LSBs}

Coppersmith's method also works if the lower half of the bits of one of the primes is known.
If $\check{p}$ is the integer corresponding to the $m$ least significant bits of $p$, then we want to find a small
root $\hat{p}$~(mod~$p$) of 
\( 2^m \cdot x + \check{p} \).
In contrast with the previous example, the factor $2^m$ has been ``removed''
from the high bits in order to reduce the size of $\hat{p}$.
Supposing that $N$ has $2k$ bits and $p$ has $k$ bits, then
$\hat{p}$ will be an integer at most $2^{k-m}$.

Coppersmith's method requires the polynomial $f$ to be monic (have a leading coefficient of 1), so to enforce this we multiply
the polynomial by $2^{-m} \bmod N$ (which exists as $N$ can be assumed to be odd).  Thus, we set
\begin{equation} f(x) = x + (2^{-m}\check{p} \bmod N) \label{eq:def_f} \end{equation}
and can apply Coppersmith's method on this $f$ to find the small root $\hat{p}$ of $f(x)$ modulo~${p}$. 
Similar to above, Coppersmith uses lattice reduction to find a polynomial with $\hat{p}$ as an \emph{integer} root;
if this polynomial has multiple integer roots, then for each root $x_0$
one should check if $\gcd(f(x_0), N)$ reveals $p$
(or more simply, if $2^m x_0+\check{p}$ divides~$N$).

\section{Previous work}\label{sec:prevwork}

The integer factorization problem has long been proposed as
a way of generating hard SAT instances;
see for example \citeauthor{Cook1997}~\cite{Cook1997}.
As noted by \citeauthor{Hamadi2013}~\cite{Hamadi2013},
the factoring problem gives particularly intriguing
instances for SAT, as the factoring problem is not expected
to be NP-hard (the decision version of the problem being
in both NP and co-NP~\cite{Pratt1975}) and is therefore a candidate for
an ``NP-intermediate'' problem lying between P and NP-hard---a
class about which little is known.

Even though factoring is unlikely
to be NP-hard, the SAT instances produced---at least
using straightforward multiplication circuits---seem
difficult, a fact confirmed by a number of independent
computational experiments~\cite{Paper1,Eriksson2014,lunden2015factoring,Asketorp2014}.
In 2022, \citeauthor{Mosca2022}~\cite{Mosca2022} reported on the state-of-the-art
for integer factorization via SAT solving
and concluded that even a quantum SAT solver would likely be slower than the best
classical algebraic methods.

Heninger and Shacham~\cite{HS09} investigate reconstructing RSA private keys with small public exponent 
from partial knowledge of the random bits of the private key.
They give a ``branch and prune'' algorithm that with high probability efficiently breaks an RSA key
given a random 27\% of bits of its private key.  Here the private key consists of both prime
factors $p$ and~$q$, the decryption exponent~$d$, as well as the two integers $d\bmod(p-1)$ and $d\bmod(q-1)$.
The analysis is heuristic, but they provide experimental evidence that their approach is effective in
practice.  Their approach makes significant use of the bits of $d$, $d\bmod(p-1)$, and $d\bmod(q-1)$
in order to limit the amount of branching.
If this extra information is not available, the approach still succeeds with
high probability if 57\% of the bits of $p$ and~$q$ are known,
or if 42\% of the bits of $p$, $q$, and~$d$ are known. 

In 2013, \citeauthor{Patsakis2013RSAPK}~\cite{Patsakis2013RSAPK} utilized SAT solvers and the encoder
ToughSat~\cite{ToughSAT} to reconstruct RSA private keys
with some partial key exposure and having a fixed public exponent of $e=3$.
He assumes the exposure was either on the bits of $p$ and~$q$ alone,
or on the bits of $p$, $q$, and~$d$.
With this information, he created SAT instances that when solved would determine the factors $p$ and~$q$.

\section{SAT + Coppersmith Approach}\label{sec:enc}

In this section we describe our hybrid SAT + computer algebra system (CAS) approach, first beginning with a basic SAT encoding
in Section~\ref{sec:sat_enc}, a description of the programmatic interface with Coppersmith's method in Section~\ref{sec:prog_enc},
and finally in Section~\ref{sec:d_enc} we describe an encoding for factoring low exponent RSA moduli
that can exploit leaked bits of the decryption exponent~$d$.

\subsection{SAT encoding}\label{sec:sat_enc}

Converting an instance of the factorization problem to a SAT instance is straightforward, as
multiplication circuits can be converted to SAT formulae by operating directly on the bit-representation of the integers.
For example, say we are forming the instance of encoding $N = p \cdot q$
where $p$ and~$q$ are known to be two integers of bitlength $k$.
We represent $p$ and~$q$ as bitvectors $[p_0,\dotsc,p_{k-1}]$ and $[q_0,\dotsc,q_{k-1}]$
and generate a multiplier circuit $\MULT$
computing the bits of the 
the product of $p$ and~$q$ from $p_0$, $\dotsc$, $p_{k-1}$ and~$q_0$, $\dotsc$, $q_{k-1}$.
The $\MULT$ circuit is constructed from chaining together full and half adder circuits,
and then the entire circuit is
converted into CNF by using the Tseytin transformation~\cite{Tseitin1968}
which introduces new Boolean variables representing the output of each gate in the $\MULT$ circuit.
For example, suppose $x$ and $y$ are the inputs to a half adder.
The Tseytin transformation
introduces a new variable~$s$ (denoting the $\mathbb{F}_2$-sum of $x$ and $y$) via $s\liff(x\oplus y)$,
and a new variable~$c$ (denoting the carry of $x$ and $y$) via $c\liff(x\land y)$.
The output bits of the $\MULT$ circuit are set to match the $2k$ bits of $N$ using $2k$ unit clauses
(or in some cases $N$ has $2k-1$ bits).
Similarly, any known bits of $p$ and $q$ are also added to the SAT instance as unit clauses (clauses of length 1).
The solver uses these unit clauses to simplify the SAT instance and improve the efficiency of the solving process.

Some simple optimizations are also encoded.  For example, $p$ and~$q$ must be odd or the problem is trivial,
so we fix the low bits $p_0$ and $q_0$ to true with unit clauses.
Similarly, since both $p$ and $q$ are assumed to be of bitlength $k$,
we fix also both high bits $p_{k-1}$ and $q_{k-1}$ to true.
Our instances were generated using the encoder of \citeauthor{CNFGen}~\cite{CNFGen}, which represents $p$ and~$q$ using
$2k-1$ and $k$ variables respectively.
However, we assign the high $k-1$ bits of~$p$ to false since we only encoded factorization problems
with $p$ and $q$ of equal bitlength.

\subsection{Coppersmith's method in programmatic SAT}\label{sec:prog_enc}

\begin{figure*}
  \centering
  \includegraphics[width=0.66\linewidth]{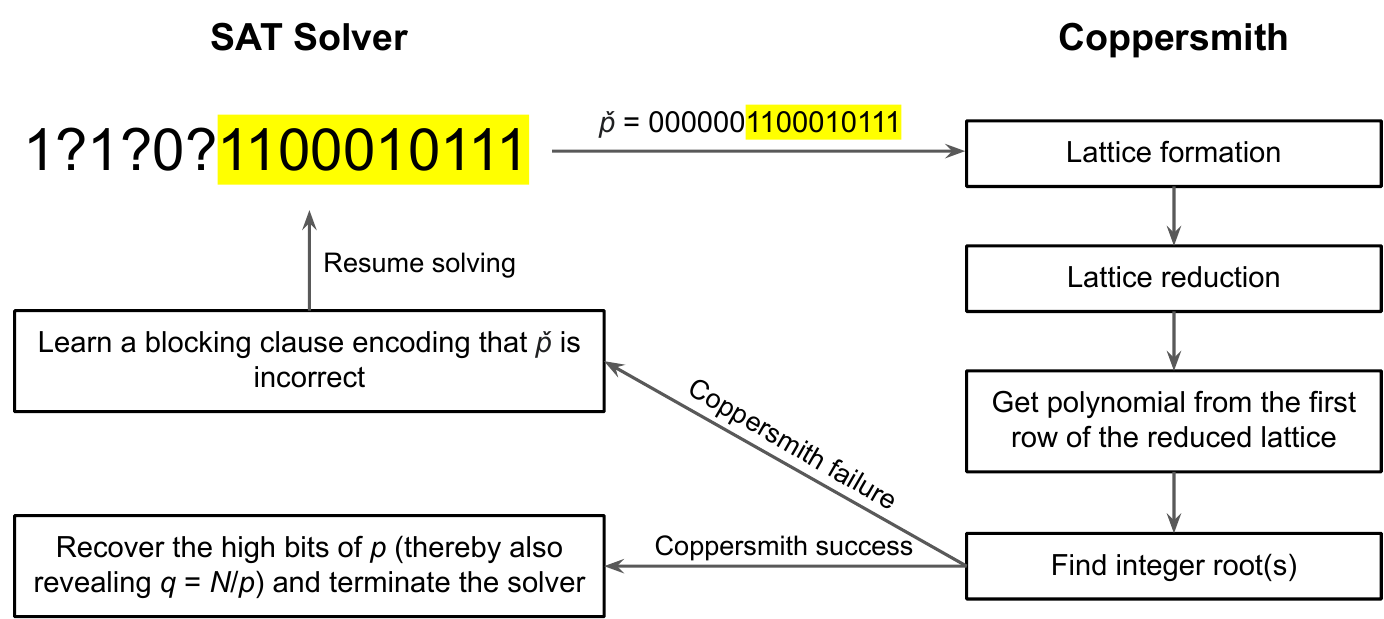}
  \caption{\boldmath A diagram outlining our SAT+CAS method for the factorization problem. By default, Coppersmith's method is invoked
  whenever the lowest $\mathord\approx60\%$ of the bits of $p$ are assigned.
  If the low bits of~$p$ were set correctly, then Coppersmith's method reveals the high bits of $p$ and the solver
  terminates.  If the low bits were set incorrectly, then Coppersmith's method fails and proves that this configuration of
  low bits cannot be extended to a solution.  In this case, a ``blocking clause'' is learned telling
  the solver to backtrack and try a new bit assignment.}
  \label{fig:working}
\end{figure*}

A ``programmatic'' SAT solver calls a custom piece of code
whenever the solver has a partial assignment that cannot
be simplified any further by unit propagation~\cite{Ganesh2012}.
In our case, the intuition behind calling Coppersmith's method is that we can use it to
test when a partial assignment can be extended to a complete assignment without requiring
the SAT solver to actually search for the extension itself.
In our experiments, the most effective strategy was to
call Coppersmith's method from within the SAT solver whenever the solver's current partial assignment
has assigned values to enough of the \textbf{\emph{low}} bits of one of the prime factors
(see Section~\ref{sec:comparison} for why using the low bits is advantageous).
Coppersmith's method can be used when the lowest $50\%+\epsilon$ bits of~$p$ are known, but
as~$\epsilon$ decreases the required lattice dimension increases, and this slows down
lattice reduction.  In practice,
we use Coppersmith's method when at least $60\%$ of the lowest bits of~$p$ are known in order to limit
the overhead from the lattice reduction.
This compromise worked well for the size of $N$ that we used in our experiments, though
it is likely for all $\epsilon>0$ there would be some
sufficiently large $N$ for which it would be worth it to call Coppersmith's method
using $50\%+\epsilon$ bits of $p$.

Following \cite{Bos_Stam_2021,DeMicheli2024}, 
a lattice of dimension 5 is sufficient to recover the unknown bits
when more than $\mathord\approx60\%$ of the lowest bits are known of one of the factors.
The polynomials used to form the lattice are $N^2$, $Nf(x)$, $f(x)^2$, $xf(x)^2$, and $x^2f(x)^2$,
where $f$ is defined as in~\eqref{eq:def_f} and taking $X\coloneqq N^{1/5}\hspace{-1pt}/4$.
The number of unknown bits cannot exceed $\log_2 X$, so if $p$ and $q$ have $k$ bits then one needs $m$,
the number of \emph{known} bits, to be larger than $k-\log_2 X\approx k-2k/5=3k/5$
or about $60\%$ of the bitlength of~$p$.
Note that if $\hat{p}$ denotes the $\lfloor\log_2 X\rfloor$ high bits of~$p$,
then $\hat{p}<X$ and 
by construction of $f$ we have $f(\hat{p})\equiv 0 \pmod p$.
Once the lattice is formed, we perform LLL lattice reduction and finally find the integer roots (if any)
of the polynomial $f_\text{red}$ associated to the first row of the reduced basis.
If the $m$ low bits of~$p$ used to construct $f$ in~\eqref{eq:def_f} were correct, then
the integer roots of $f_\text{red}$
include the mod-$p$ roots of $f$ of absolute value at most~$X$ (see \cite[ch.~19]{Galbraith}
for details).  In other words, $\hat{p}$ is among the integer roots of $f_\text{red}$
if $\check{p}$ was set correctly in $f$.

For each small integer root $x_0$ of $f_\text{red}$ returned by Coppersmith's method, a validation step is executed.
If $\gcd(f(x_0), N)$ is nontrivial, then
the procedure concludes successfully with a factorization of $N$.
However, in cases where no roots provide a factor of~$N$,
a ``blocking clause'' is added to the
SAT solver's learned clause database.
The blocking clause encodes that the combination of the low bits passed to Coppersmith's method was erroneous
by stating that at least one of the bits must change from its current assigned value.
For example, suppose
Coppersmith's method is applied to an 8-bit prime with the assignment $p = \text{\texttt{{???}10011}}$ and fails.
Then the conflict clause will be $\lnot p_4\lor p_3\lor p_2\lor\lnot p_1\lor\lnot p_0$
where the bits of $p$ (from low to high) are represented by the variables $p_0$, $\dotsc$, $p_7$.
The solver incorporates this knowledge as a learnt clause and
immediately backtracks to explore alternative bit combinations.
Figure~\ref{fig:working} visually depicts how the technique works.

\subsection{Low public exponent RSA encoding}\label{sec:d_enc}

We also considered a special case of the factorization problem,
namely, the problem of factoring an RSA modulus $N$ with a public exponent of $e=3$
(implying that both $p-1$ and $q-1$ are not divisible by 3).
In such a case it is possible to derive~\cite{boneh1999twenty}
the equation
\begin{equation} 3d + 2(p + q) = 2N + 3 \label{eq:enc_d} \end{equation}
where $d$ is the decryption exponent.  Moreover, we can approximate
$d$ by $\tilde{d} = \lfloor (2N+3)/3 \rfloor$ because $2(p+q)$
is relatively small compared to $N$.
Indeed, if $p\geq q$ and both factors have $k$ bits then
$q\leq\sqrt{N}$ and $p<2\sqrt{N}$, so $p+q<3\sqrt{N}$.
As pointed out by \citeauthor{Boneh1998}~\cite{Boneh1998},
one can derive
\begin{equation} 0 \leq \tilde{d}-d < 3\sqrt{N}, \label{eq:bounds} \end{equation}
and they remark
\begin{center}
\emph{``It follows that\/ $\tilde{d}$ matches $d$ on the
$n/2$ most significant bits of $d$.''}
\end{center}
Similarly, \citeauthor{HS09}~\cite{HS09} remark that
$\tilde{d}$ \emph{``agrees with $d$ on their $\lfloor n/2\rfloor-2$ most
significant bits''}.\footnote{In both of these quotes~$n$ denotes the bitlength of $N$.}
Surprisingly, both claims are false as
adding even a small difference $\tilde{d}-d<3\sqrt{N}$
to $d$ can in some cases cause a cascade of carries
changing bits well into in the upper-half of $d$.
%
For example, when $N=827\cdot953$, one has
$d=2^{19}-53$ and $\tilde{d}=2^{19}+1133$
which share \emph{no} high bits
(as bitstrings of length 20). 
We noticed this oversight when we attempted to set the high bits
of $d$ to match the high bits of~$\tilde{d}$ (computed
from $\lfloor 2N/3+1 \rfloor$) and in some cases the resulting
instances were shown to be unsatisfiable by the SAT solver.
We resolved this by using the following lemma, which
also gives a slightly stronger version of~\eqref{eq:bounds}, replacing the constant $3$ with $\sqrt2$.
\begin{lemma}\label{lem:msb_d}
Let\/ $N=pq$ be an\/ $n$-bit RSA modulus where\/ $p$ and\/~$q$ have the same bitlength,
suppose\/ $d$ is the decryption exponent for encryption exponent\/ $e=3$,
and set\/ $\tilde{d}=\lfloor 2N/3+1 \rfloor$.  Then
\begin{enumerate}
\item[(a)] $0 \leq \tilde{d}-d < \sqrt{2N}$.
\item[(b)] Write\/ $\tilde{d}$ and\/ $\tilde{d}-\lfloor \sqrt{2N}\rfloor$ as bitstrings
of length\/ $n$, and suppose the upper\/ $l$ bits of
the bitstrings match.
Then the upper\/ $l$ bits of\/ $d$'s bitstring of length\/~$n$ match those of\/ $\tilde{d}$.
\end{enumerate}
\end{lemma}
\begin{proof}
Without loss of generality suppose $q\leq p < 2q$, so that $pq<2q^2$
(i.e., $q>\sqrt{N/2}$) and $q^2\leq pq$ (i.e., $q\leq \sqrt{N}$).
Then $p+q=N/q+q$ and $F(q)\coloneqq N/q+q$ is monotonically decreasing
over $q\in\bigl(\sqrt{N/2},\sqrt{N}\bigr]$, so $p+q<F\bigl(\sqrt{N/2}\bigr)=3\sqrt{2N}/2$.
Using~\eqref{eq:enc_d}
we have
\begin{equation*} 0 \leq \tilde{d}-d \leq 2(p+q)/3 < 2F\bigl(\sqrt{N/2}\bigr)/3 = \sqrt{2N} \label{eq:imp_bounds} \end{equation*}
which is the inequality in (a).

The inequality in (a) is equivalent to
$d\in\bigl(\tilde d-\sqrt{2N},\tilde d\bigr]$.
By assumption, the bitstrings of the lowest and highest integers in this range
have~$n$ bits and share the same $l$ high bits.
The only way this can happen is if \emph{all} bitstrings of integers in
this range all share the same $l$ high bits, including~$d$.
Otherwise, if we want the high bit (i.e., the bit
of index $n-1$)
to match in the lowest and highest integers but \emph{not}
with some integer in the range we would need the range to contain
at least $2^n$ integers which it does not.
\end{proof}

Equation~\eqref{eq:enc_d} can be encoded in SAT using a binary adder on the terms
of the left-hand side, reusing the variables for the bits of $p$ and~$q$
and introducing new variables for the bits of~$d$.  The output bits
of the binary adder are then set to the binary representation of $2N+3$.
The upper bits of $d$
are fixed to those of $\tilde{d}$ using unit clauses
(with the number of bits fixed determined by Lemma~\ref{lem:msb_d}).
Any bits of~$d$ that are leaked can also be added as unit clauses.

\section{Results}\label{sec:results}

\begin{figure*}
  \centering
  \begin{subfigure}{0.49\textwidth}
      \centering
      \includegraphics[width=\linewidth]{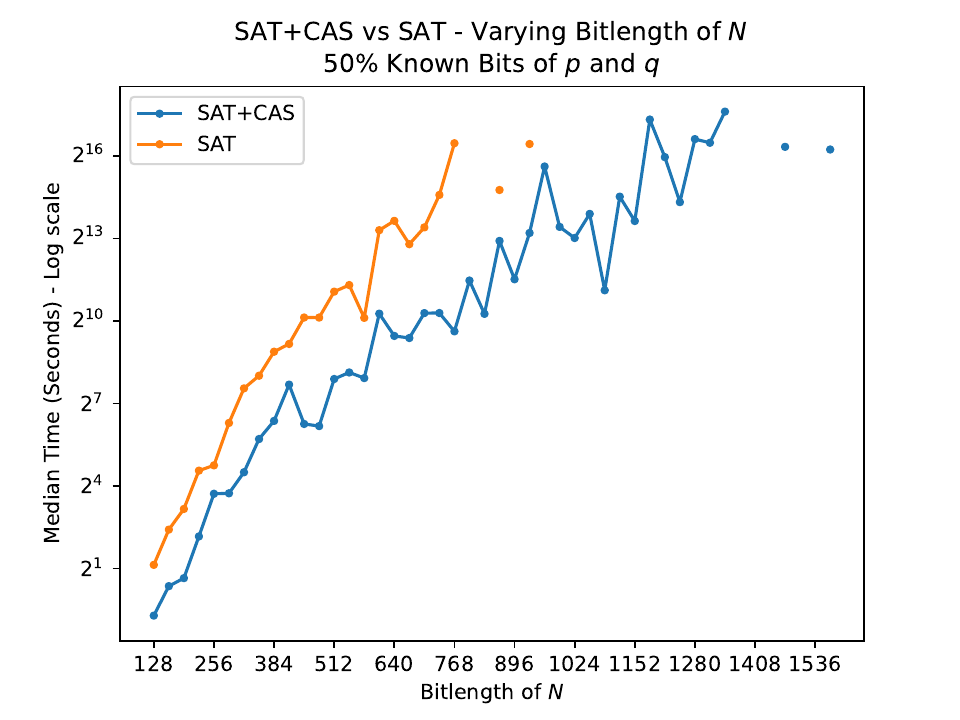}
      \caption{}
      \label{fig:varyingN}
    \end{subfigure}
    \begin{subfigure}{0.49\textwidth}
      \centering
      \includegraphics[width=\linewidth]{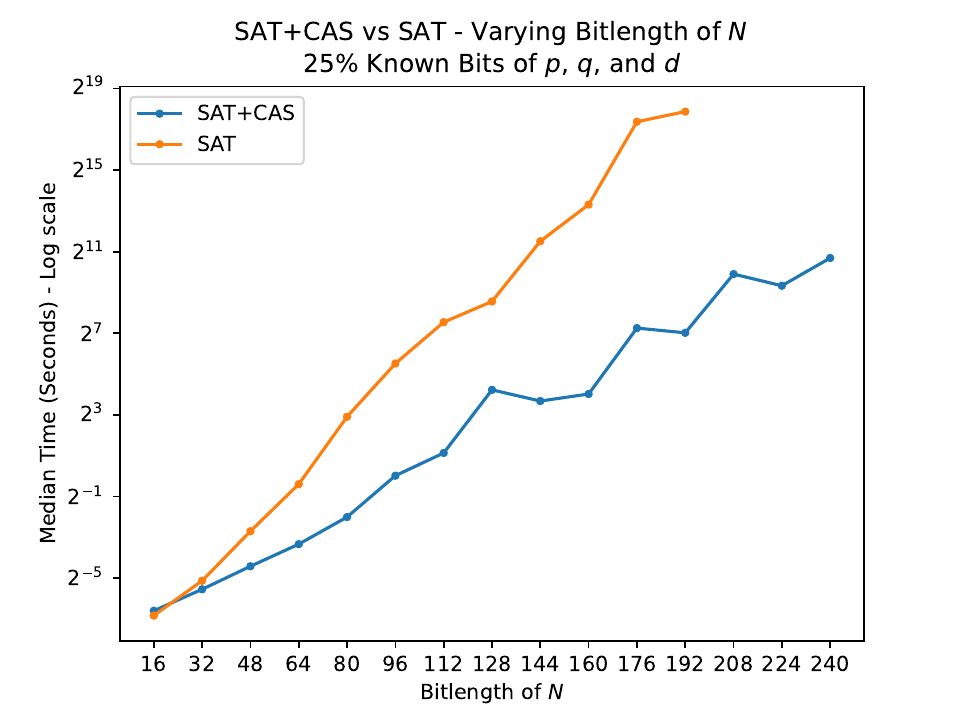}
      \caption{}
      \label{fig:varyingN_withd}
    \end{subfigure}
    \begin{subfigure}{0.49\textwidth}
      \centering
      \includegraphics[width=\linewidth]{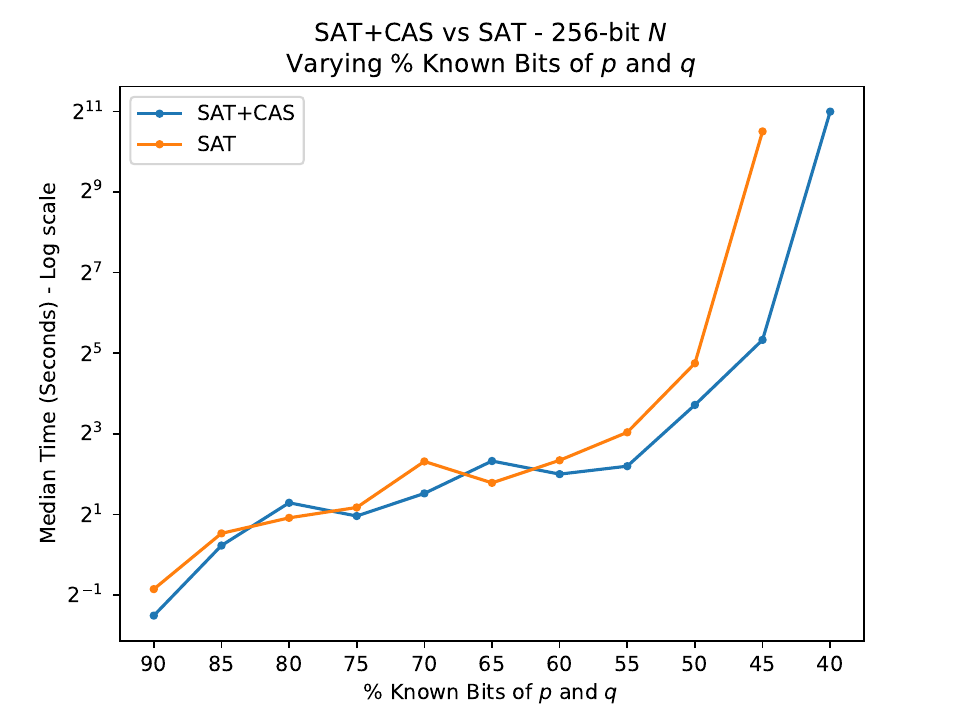}
      \caption{}
      \label{fig:varying_percent}
    \end{subfigure}
    \begin{subfigure}{0.49\textwidth}
      \centering
      \includegraphics[width=\linewidth]{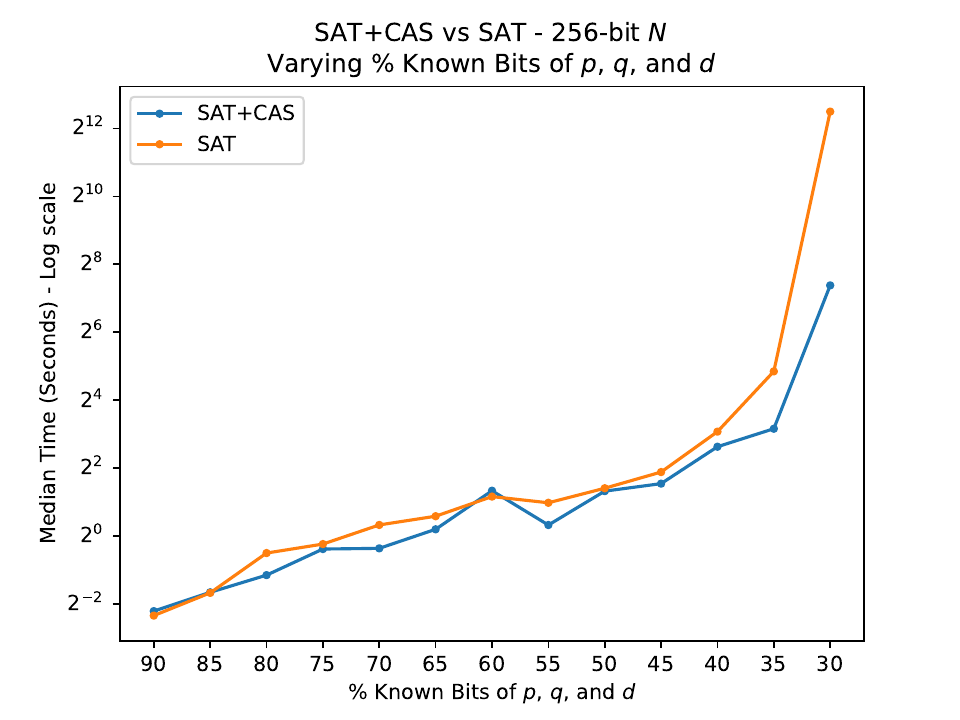}
      \caption{}
      \label{fig:varying_percent_withd}
    \end{subfigure}
%
%
%
    \caption{\boldmath The upper plots compare the median running time across different sizes of~$N$.
    The left plots summarizes instances with random bits of $p$ and $q$ leaked, while the right plots also leak bits of the decryption exponent~$d$.
    The lower plots compare the median running time for a 256-bit $N$ across a varying percentage of known bits.
    All instances were run with a timeout of 3 days, so the lack of a point on the graph indicates the median time was over 3 days.
    %
    %
    All plots are given on a logarithmic scale.}\label{fig:results}
\end{figure*}

To create our SAT instances, we employ the \emph{CNF Generator for Factoring Problems} by \citeauthor{CNFGen}~\cite{CNFGen} 
using the ``$N$-bit'' adder type and the ``Karatsuba'' multiplier type, as we found those to be the most effective.
These instances undergo subsequent enhancements through the integration of supplementary clauses, as detailed in Section~\ref{sec:sat_enc}.
Our SAT solver that calls Coppersmith and the scripts we used to perform our experiments
are available on a public GitHub repository at \url{https://github.com/yameenajani/SAT-Factoring}. 

\subsection{Solving method}

The instances were solved using a programmatic version of MapleSAT~\cite{DBLP:conf/sat/LiangGPC16}
available as a part of the MathCheck project~\cite{bright2016mathcheck2}. 
The version of Coppersmith's algorithm used is a custom implementation in C++.
The GMP 
library~\cite{GMP} was used to form the lattice and it was reduced using the fplll library~\cite{fplll}.
The formation of the polynomial from the reduced basis and its factorization is done using FLINT~\cite{flint}.
All experimentation took place on Compute Canada's Cedar cluster
with each instance solved on a single Intel E5-2683 Broadwell CPU core running at 2.1~GHz and allocated 4~GiB of memory.

Each experimental iteration commences with the generation of an appropriately sized modulus~$N$ using a SageMath~\cite{sagemath} script.
The modulus is then passed as input to the CNF Generator, which in turn generates the requisite CNF and delivers it in the DIMACS SAT file format.
To this file, we append the unit clauses that specify known bits (selected uniformly at random) of $p$ and~$q$.
The percentage of known bits is fixed and given as an argument to the script.
There is also an option to encode the high bits of~$d$
assuming $N$ is a low public exponent modulus (i.e., its prime factors
are not congruent to 1 mod~3)
using the encoding described in Section~\ref{sec:d_enc}.

\subsection{Summary of results}\label{sec:results_summary}

We tested our method on random semiprime factorization problems
where both $p$ and $q$ are not congruent to 1 mod 3,
both with and without the ``low public exponent'' encoding
described in Section~\ref{sec:d_enc}.
In each problem, we leaked a selection of the bits
of $p$ and~$q$ chosen uniformly at random.  For the low public exponent RSA problems, we
leaked a selection of the bits of $d$ chosen uniformly
at random in addition to the high bits of $d$ that could be analytically derived
using Lemma~\ref{lem:msb_d}.
We generated 15 random keys for varying bitsizes of~$N$
ranging from 16 bits to 1728 bits,
and for each key we randomly leaked a percentage of the bits
of the private keys ranging from 90\% to 25\% (in increments of 5\%).
We ran the solver on the SAT factorization problem produced from each key
and plotted the resulting median times across
several different types of problems in Figure~\ref{fig:results}.

The first set of experiments fixes the number of known bits of $p$ and~$q$,
and increases the bitlength of $N$ until the instances become too hard to solve.
As indicated in Figure~\ref{fig:varyingN}, a 768-bit $N$ with 50\% leaked bits
takes a pure SAT approach a median of 90,521 seconds to factor,
while the SAT+CAS approach factors it in a median of 789 seconds.
In these instances each call to Coppersmith
took about 0.005 seconds and Coppersmith was called a median of 490 times
and used a median of 0.5\% of the total running time.
The instances containing leaked bits of~$d$ were significantly easier to solve,
so we repeated the same experiments
but only leaked 25\% of the bits of $p$, $q$, and~$d$.
The results are indicated in Figure~\ref{fig:varyingN_withd}.
For example, a 192-bit $N$ with 25\% leaked bits
takes a pure SAT approach a median of 239,992 seconds to factor~$N$,
while the SAT+CAS approach factors~$N$ in a median of 130 seconds.
In these instances each call to Coppersmith took about 0.002 seconds
and Coppersmith was called a median of 6466 times and used a median
of 13.5\% of the total running time.

The results shown in Figure~\ref{fig:varying_percent} fix the size of $N$ to 256 bits and vary the
percentage of known bits of $p$ and~$q$.
When a large number of bits are known (at least 50\%)
both the SAT and SAT+CAS approaches perform relatively well.  In fact,
when the percentage of known bits is higher than 60\%, the simpler
pure SAT approach can even outperform the more involved SAT+CAS approach.
However, the SAT+CAS approach clearly scales better.
For example, with 45\% leaked bits, the pure SAT solver
factors~$N$ in a median of 1452 seconds,
while the SAT+CAS solver factors $N$ in a median of 40 seconds.
With 40\% leaked bits, the SAT+CAS solver factors $N$ in a median
of 2042 seconds, while the median time of the
pure SAT approach does not complete after 259,200 seconds
(the solver timeout was set to 3 days).
The experiments shown in Figure~\ref{fig:varying_percent_withd} are similar, but
random bits of $d$ are also provided to solver.  In this case,
for all percentages down to 30\% the median instance was
solved within the timeout for both the SAT and SAT+CAS solvers.
However, with 30\% known bits the median SAT time was 5778 seconds,
while the median SAT+CAS time was 167 seconds.
For 256-bit $N$, each call to Coppersmith took about 0.002 seconds.
When 45\% of the bits of $p$ and $q$ were known, Coppersmith was called
a median of 793 times, and when 30\% of the bits of $p$, $q$, and~$d$ were known, Coppersmith
was called a median of 1165 times.

\subsection{Comparison with other approaches}\label{sec:comparison}

Our results show that the SAT+CAS method outperforms
not only a SAT-only approach, but also a brute-force
approach, even if it uses Coppersmith's method.
For example, with 50\% leaked bits of $p$ and $q$, a 512-bit $N$
can be factored by the SAT+CAS solver in a median of 237 seconds,
but a Coppersmith + brute-force approach would need to determine
values for around 64 unknown bits in the lower half of~$p$
before Coppersmith could be applied---much more expensive
given the speed of Coppersmith.
Additionally, the SAT+CAS solver will also be much more efficient
than the number field sieve on the specific problem of factoring
a 512-bit $N$ with 50\% leaked bits, given that
factoring a 512-bit~$N$ with the number field sieve takes around 2770 CPU hours on
Amazon's Elastic Compute Cloud (EC2) service~\cite{Valenta2017}.

We also tried comparing our implementation with the ``branch and prune'' implementation
of \citet{HS09}.  Their approach starts from the low bits of the private key and moves towards
the high bits incrementally, enumerating all possibilities for the lowest $i$ bits by branching on
(and pruning branches when possible) all possibilities for the lowest $i-1$ bits.
Their approach is very effective when the number of branches
does not grow too large, which in practice happens if enough random bits are known of the
private key.  For example, with 45\% randomly leaked bits of $p$ and $q$
for a 256-bit $N$, their implementation required at most 640,000
branches across 15 random trials and in each case~$N$ was factored in under 5 seconds.

When the percentage of known bits dropped too low, their implementation suffered from an
exponential blowup in the number of branches resulting in excessive memory usage.
For example, with 25\% leaked bits of $p$, $q$, and~$d$ and
a 192-bit $N$, across 31 random trials
their implementation required a median of 137 million branches and 46.2~GiB
of memory, taking a median of 491 seconds to factor~$N$ on an Intel i7 CPU running at 2.8~GHz.
A SAT+CAS solver running on the same machine and solving
31 instances with the same proportion of leaked bits
used 
a median of 69 seconds and 87~MiB of memory.
On such instances, the median number of branches for the lowest 60 bits of $p$, $q$, and~$d$
was 213,161 using Heninger and Shacham's code---indicating that the SAT solver, which called
Coppersmith after the lowest 60 bits of~$p$ are set and used a median
of 15,976 Coppersmith calls, reduces 
the number of possibilities for the lowest 60 bits over
a pure ``branch and prune'' approach.

An examination of the low bits of $p$ and $q$ in the partial assignments explored by the SAT solver show that
the solver only explores partial assignments satisfying
Heninger and Shacham's pruning constraints.  In other words,
the solver does not waste time exploring branches that Heninger and Shacham
prune---essentially, the SAT solver incorporates the pruning conditions
without being explicitly told them.  This is likely why calling Coppersmith using the low
bits is much more effective than using the high bits, as the solver
avoids exploring many possibilities for the low bits.  Although there has also been work
done on pruning constraints using the high bits~\cite{Sarkar2013}, these constraints
are more involved and require mathematical context that the solver likely cannot
derive from the SAT encoding alone.  However, in the future
a programmatic SAT solver could potentially incorporate pruning on the high bits.

\section{Conclusion}

In this work we demonstrate the performance of SAT solvers on integer factorization problems
can be dramatically improved by calling a computer algebra system (CAS) during solving
in order to reveal algebraic structure unknown to the solver.  Specifically, our
programmatic SAT+CAS solver calls Coppersmith's method when a significant portion of
the bits of the prime factors have been assigned.  Coppersmith's method is then able to efficiently
(a) uncover the remaining unknown bits; or (b) tell the solver that the current bit assignment is incorrect
and have the solver backtrack immediately.  The latter is the typical case and our results
demonstrate that even with the overhead of querying a CAS the ability to backtrack early causes
the solver to factor integers significantly more efficiently, with a speedup factor that
in practice is exponential in the bitlength of $N$---see Figures~\ref{fig:varyingN} and~\ref{fig:varyingN_withd}.

Although there has been much recent work on adding algebraic reasoning into a SAT solver, to
our knowledge the algebraic information used in our work has previously only been exploited by
computer algebra systems and not SAT solvers.

\paragraph*{Author note}
A preliminary version of this work appeared as an extended abstract
in the \emph{2023 SC-Square workshop}~\cite{DBLP:conf/scsquare/AjaniB23}.  Regrettably, the initial timings
reported in the extended abstract are not trustworthy and should be disregarded.
We regret the error.

\begin{acks}
We thank the anonymous reviewers whose detailed comments improved this paper.
\end{acks}

\bibliographystyle{ACM-Reference-Format}
\bibliography{issac24}

\end{document}